\documentclass[10pt,final,letterpaper,twocolumn,journal]{IEEEtran}
\usepackage{cite}
\usepackage{graphicx,color,epsfig,rotating}
\usepackage{amsfonts,amsmath,amssymb}
\usepackage{algorithm,algorithmic}
\usepackage{url}
\usepackage{color}
\usepackage{epsfig}
\usepackage{multicol}
\usepackage{multirow}

\newtheorem{definition}{Definition}
\newtheorem{theorem}{Theorem}
\newtheorem{lemma}{Lemma}

\newcommand{\rank}{\text{rank}}

\begin{document}
\title{A Repair Framework for Scalar MDS Codes}
\author{
\IEEEauthorblockN{Karthikeyan Shanmugam\IEEEauthorrefmark{1} \IEEEmembership{Student~Member,~IEEE}, Dimitris S. Papailiopoulos\IEEEauthorrefmark{1} \IEEEmembership{Student~Member,~IEEE}, Alexandros G. Dimakis\IEEEauthorrefmark{1} \IEEEmembership{Member,~IEEE}, and Giuseppe Caire\IEEEauthorrefmark{2} \IEEEmembership{Fellow,~IEEE}}

\IEEEauthorblockA{ \IEEEauthorrefmark{1}Department of Electrical and Computer Engineering \\
The University of Texas at Austin, Austin, TX-78712 \\
\texttt{\{karthiksh,dimitris\}@utexas.edu,dimakis@austin.utexas.edu}} 

\IEEEauthorblockA{ \IEEEauthorrefmark{2}Department of Electrical Engineering \\
University of Southern California, Los Angeles, CA-90089 \\
\texttt{caire@usc.edu}}
}
\maketitle

\begin{abstract}
Several works have developed {\it vector-linear} maximum-distance separable (MDS) storage codes that minimize the total {\it communication cost} required to repair a single coded symbol after an erasure, referred to as repair bandwidth (BW). Vector codes allow communicating fewer sub-symbols per node, instead of the entire content. This allows non trivial savings in repair BW. In sharp contrast, classic codes, like Reed-Solomon (RS), used in current storage systems, are deemed to suffer from {\it naive repair}, i.e. downloading the entire stored message to repair one failed node. This mainly happens because they are \textit{scalar-linear}. 

In this work, we present a simple framework that treats scalar codes as vector-linear. In some cases, this allows significant savings in repair BW. We show that vectorized scalar codes exhibit properties that simplify the design of repair schemes. Our framework can be seen as a finite field analogue of real interference alignment.

Using our simplified framework, we design a scheme that we call {\it clique-repair} which provably identifies the best linear repair strategy for any scalar $2$-parity MDS code, under some conditions on the sub-field chosen for vectorization. We specify optimal repair schemes for specific $(5,3)$- and $(6,4)$-Reed-Solomon (RS) codes. 
Further, we present a repair strategy for the RS code currently deployed in the Facebook Analytics Hadoop cluster
that leads to $20\%$ of repair BW savings over {\it naive repair} which is the repair scheme currently used for this code.
\let\thefootnote\relax\footnote[1]{ This paper was presented in part at \textit{50th Annual Allerton Conference on Communication, Control and Computing} 2012 \cite{shanmugam2012repair}. This research was partially supported by NSF Awards 1055099, 1218235 and research gifts by Google, Intel and Microsoft.}%
\end{abstract}

\section{Introduction}

Large-scale distributed storage systems employ erasure coding to offer data reliability against hardware failures. Typically, the erasure codes employed are $(n,k)$ MDS (maximum distance separable) codes. An important property that ensures data reliability against failures, is that encoded data from any $k$ nodes suffice to recover the data stored. However, a central issue that arises in coded storage is the {\it Repair Problem}: how to maintain the encoded representation when a \textit{single} node erasure occurs. To maintain the same redundancy posterior to an erasure, a new node has to join the storage array and regenerate the lost contents by downloading and processing data from the remaining storage nodes. Classic codes, like Reed-Solomon are scalar MDS codes. Currently used repair scheme for these codes is \textit{naive repair}. This involves downloading all the contents of any $k$ of the remaining nodes to reconstruct the entire file and then replacing the coded sub-symbols of a single failed node. 

	During repair process of an erasure, there are several metrics that can be optimized, namely \textit{repair bandwidth} (BW) and \textit{locality}\cite{Local1}\cite{Local2}\cite{Local3}\cite{Local4}\cite{huang2012erasure}\cite{sathiamoorthy2013xoring}. Currently, the most well understood one is the total number of bits communicated in the network, i.e. \textit{repair bandwidth} (BW). This was characterized in~\cite{DimakisGWWR:08} as a function of storage per node. 
Codes with minimum storage that offer optimal bandwidth are MDS, and are called minimum storage regenerating (MSR) codes. Building on the work in~\cite{DimakisGWWR:08}, a great volume of studies have developed MSR codes \cite{Survey,storagewiki,Tamo, PermCodes,RashmiProduct,SuhR:09,KVSKR:09,cadambe2011polynomial,PD1,rashmi2013piggy}.   

	In this paper, we deal with the following specific repair scenario for \textit{systematic} MSR codes: a file consisting of $M$ sub-symbols, over some field, is stored in $n$ nodes using an $(n,k)$ vector systematic MDS code. Every node contains $\alpha=\frac{M}{k}$ sub-symbols over the field. The first $k$ systematic nodes store uncoded sub-symbols in groups of $\frac{M}{k}$. The parity nodes contain the coded data. An MDS code can tolerate $n-k$ erasures. 
	
	Suppose one of the systematic nodes fail and this needs to replaced. For such a repair, $\beta$ sub-symbols are downloaded from every remaining parity node, through suitable linear combinations of the $\alpha$ symbols present in each parity node. From every remaining systematic node, at least $\beta$ symbols are downloaded. The downloaded symbols must be sufficient to generate the contents of the failed node through linear operations for successful repair.  According to the cut-set lower bounds of \cite{DimakisGWWR:08}, the optimum per-node download achievable by any code and repair scheme is $\beta = \frac{M}{k (n-k)} = \frac{\alpha}{n-k}$ when exactly $\beta$ symbols are downloaded from \textit{all} the remaining $n-1$ nodes for repair. It is easily seen that the optimum repair strategy for MSR codes has immense benefit over naive repair for constant rate codes and for large $n$. The key property of MSR codes that enables the non-trivial repair is that they are vector codes, i.e. data in a node is a collection of smaller sub-symbols over a field and few linear combinations from every node suffice for repair of a single failure. 

	MSR codes, with efficient encoding and decoding schemes, that meet the minimum cut-set BW bounds, derived in \cite{DimakisGWWR:08}, exist for rate $k/n \leq 1/2$ \cite{RashmiProduct}. In the high rate regime, \cite{Tamo, PermCodes,SuhR:09,KVSKR:09,PD1} have presented constructions that achieve the optimal repair bandwidth. However, the amount of \textit{subpacketization} (sub-symbols) required is exponential in the parameters $n$ and $k$. Code constructions in \cite{cadambe2011polynomial} rectify this problem by constructing high rate codes, for specific rates, that have polynomial subpacketization. For more details on regenerating codes for other scenarios, we refer the reader to the surveys \cite{Survey1}\cite{Survey2}\cite{Survey}.

	An interesting problem is developing repair strategies for {\it existing} systematic scalar linear MDS codes that are currently used in erasure coded storage systems. A major limiting issue of these codes is that they lack the fundamental ingredient of repair optimal ones: the vector-code property. Naive repair is currently the only known strategy for these codes.
	
In this work, we focus on repairing a failed systematic node of a systematic scalar linear MDS code, defined over a large extension field.  The focus is \textit{not} on designing codes that achieve the cut-set bound of $\frac{\alpha}{n-k} (n-1)$ for repair bandwidth but on analyzing the repair efficiency of existing ones. We show that any scalar linear MDS code, can be vectorized over a suitable smaller sub-field. When a systematic scalar linear MDS code is vectorized, the problem of designing the right linear combinations of stored symbols, from a surviving parity node, to be used for repair (also called as \textit{repair vector} design) can be equivalently seen as the problem of designing \textit{repair field elements}.  Instead of designing a repair vector for each equation downloaded, a field element belonging to the extension field is chosen for every repair equation of a vectorized scalar code. These field elements satisfy some linear independence constraints.  This equivalent formulation is the main technical contribution of the paper. This gives some analytical insights for repair of $2$-parity codes when vectorized over specific sub-fields. We summarize our contributions below.

{\bf Our contributions:}
In this work, we develop a framework to represent scalar linear MDS codes in a vector form, when they are constructed over extension fields. 
The vector form provides more flexibility in designing non-trivial repair strategies. We pose the problem of designing repair vectors (the best linear combinations to download) for repairing a systematic node, as a problem of designing repair field elements satisfying some algebraic linear dependence properties. Using this framework, we develop an algorithm, called clique-repair, that outputs an optimal repair scheme for a given 2-parity scalar linear MDS code, viewed as a vector code over a suitably chosen sub-field. This is based on an analytical condition, obtained through the repair field elements approach, that \textit{directly} relates the code's generator matrix entries to the repair bandwidth. We show that, for a specific $(6,4)$ Reed Solomon code, the clique repair scheme obtains nontrivial gains in terms of repair bandwidth.  For this $(6,4)$ RS code and another specific $(5,3)$ RS code, the gains can be brought close to the optimal cut-set bound of \cite{DimakisGWWR:08}, by vectorizing over a smaller sub-field. Further, we present numerical results regarding the repair of the $(14,10)$ RS code currently used in production \cite{sathiamoorthy2013xoring} by Facebook Hadoop Analytics cluster. There, we observe a $20\%$ savings in terms of repair BW compared to naive repair.

\section{Repair of MDS Storage Codes}

In this section, we first state the repair BW minimization problem for systematic vector MDS codes to clarify the implications of storing vectors per node instead of scalars. Throughout the paper, we consider the case of downloading sub-symbols from all the remaining $n-1$ nodes to repair a single failed node. We see that scalar-linear MDS codes have an inherent deficit when assuming indivisible coded symbols.

\subsection{Vector MDS Codes} \label{Sec:VectorMDS}
Let a file ${\bf x}$ be subpacketized into $M=k(n-k)\beta$ $p$-ary information symbols such that ${\bf x}\in\mathbb{F}^{M\times 1}$ and partitioned in $k$ parts ${\bf x}=\left[{\bf x}^T_1\ldots{\bf x}^T_k\right]^T$, with ${\bf x}_i\in\mathbb{F}^{\frac{M}{k}\times 1}$, where $M$ denotes the file size and $\mathbb{F}\equiv\mathbb{GF}(p)$. Here, the number of sub-symbols over $\mathbb{GF}(p)$ stored in a node is $\alpha=\frac{M}{k}=(n-k)\beta$. Let us define the \textit{degree of subpacketization} to be $\beta\in\mathbb{Z}^{+}$. We want to store this file with rate $\frac{k}{n}\le 1$ across $k$ systematic and $n-k$ parity storage units with storage capacity $\frac{M}{k}$ $p$-ary symbols each. 

The encoding is given by:

{\small

\begin{equation} \label{Eqn:vectorMDSeq}
\mathbf{y}=
\left[
\begin{array}{c}
 \mathbf{y}_1\\
  \vdots \\
   \vdots \\
   \vdots \\
  \mathbf{y}_n
\end{array}
\right]=
\left[
\begin{array}{ccc}
\mathbf{I}_{\alpha}^{1}& \ldots & \mathbf{0} \\
    \mathbf{0}         & \ddots & \mathbf{0} \\
     \mathbf{0}        & \ldots  & \mathbf{I}_{\alpha}^{k} \\
\hline
{\bf P}^{(k+1)}_1 & \ldots & {\bf P}^{(k+1)}_k\\
\vdots & \vdots&\vdots\\
{\bf P}^{(n)}_k & \ldots & {\bf P}^{(n)}_k
\end{array}
\right]
\left[
\begin{array}{c}
 \mathbf{x}_1 \\ 
 \vdots \\
 \mathbf{x}_k
\end{array}
\right]
\end{equation}
}where ${\bf P}^{(j)}_i\in\mathbb{F}^{\alpha \times \alpha}$ represents a matrix of coding coefficients used by the $j$th node ($j \geq k+1$ and hence a parity node) to ``mix'' the symbols of the $i$th file piece ${\bf x}_{i}$. $\mathbf{y}_i$ denotes the vector of coded sub-symbols stored in node $i$. $\mathbf{I}_{\alpha}^{i}$ denotes an $\alpha \times \alpha$ identity matrix. The MDS property is guaranteed if the file can be reconstructed from from any subset of size $k$ of the $n$ nodes storing the codeword $\mathbf{y}$.

\textbf{Remark:} The choice of $M$ being a multiple of $k(n-k)$ is due to the following reasons:
\begin{enumerate}
   \item The lowest per node repair bandwidth (in terms of sub-symbols over $\mathbb{GF}(p)$) possible is $\beta=\frac{\alpha}{n-k}$ sub-symbols according to the cut-set bound in \cite{DimakisGWWR:08}. 
   \item Further, for repair of all systematic vector codes occurring in this work, we assume that the number of sub-symbols that will be downloaded from every surviving parity node will be $\beta=\frac{\alpha}{n-k}=\frac{M}{k(n-k)}$. However, note that it may not be possible to download only $\beta$ symbols from each surviving systematic ones unless optimal repair is feasible for that code.  In fact, the goal of efficient repair will be to download as close to $\beta$ sub-symbols as possible from each surviving systematic node.
\end{enumerate}   

\subsection{Repair Vector Design Problem} \label{Sec:RepairDes}
Let $\left[k \right]$ denote the set $\{1,2,3 \ldots k\}$. To maintain the same redundancy when a single systematic node $i\in[k]$ fails, a repair process takes place to regenerate the lost data in a {\it newcomer} storage node.
This process is carried out as linear operations on the content of the $n-1$ remaining nodes, namely, each parity node $j\in\{k+1 \ldots n\}$ sends data of size $\beta=\frac{M}{k(n-k)}$ (i.e., $\beta$ equations) to the newcomer in the form of linear equations:
{\small
\begin{align}
{\bf d}_i^{(j)}&=\left({\bf R}^{j}_i\right)^T\left(\left({\bf P}^{(j)}_1\right){\bf x}_1+\cdots+\left({\bf P}^{(j)}_k\right){\bf x}_k\right)\nonumber\\
&=\left[\left({\bf R}^{j}_i\right)^T{\bf P}^{(j)}_1 \cdots \left({\bf R}^{j}_i\right)^T {\bf P}^{(j)}_k\right] {\bf x},
\end{align}
}where ${\bf R}_i^{j}\in\mathbb{F}^{(n-k)\beta\times \beta}$ is a \textit{repair matrix}, which is to be designed.
In the same manner, all parity nodes proceed in transmitting a total of $\frac{M}{k}$ linear equations (i.e., the size of what was lost) to the newcomer, which eventually receives the following system  of linear equations
{\small
\begin{align}\label{Eqn:repaireq}
{\bf d}_i \hspace{-0.1cm}&=
\hspace{-0.1cm}\underbrace{\left[
\begin{array}{@{}c@{}}
\left({\bf R}^{k+1}_i\right)^T{\bf P}^{(k+1)}_{i}\\
\vdots\\
\left({\bf R}^{n}_i\right)^T{\bf P}^{(n)}_i
\end{array}
\right]\hspace{-0.1cm}{\bf x}_i}_{\text{useful data}}\hspace{-0.1cm}+\hspace{-0.35cm}\sum_{u=1,u\ne i}^k\hspace{-0.1cm}
\underbrace{\left[
\begin{array}{@{}c@{}}
\left({\bf R}^{k+1}_i\right)^T{\bf P}^{(k+1)}_{u}\\
\vdots\\
\left({\bf R}^{n}_i\right)^T{\bf P}^{(n)}_u
\end{array}
\right]\hspace{-0.1cm}{\bf x}_u}_{\text{interference by ${\bf x}_u$}},
\end{align}
}where ${\bf d}_i\in\mathbb{F}^{\frac{M}{k}}$. 
Solving for ${\bf x}_i$ is not possible due to the $(k-1)$ additive {\it interference} components in the received equations.
To retrieve the lost piece of data, we need to ``erase'' the interference terms by downloading additional equations from the remaining $k-1$ systematic nodes and the resulting system has to be full-rank.
To erase the interference generated by the undesired symbols ($\mathbf{x}_u$), we need to download from systematic node $u$ the minimum number of equations that can re-generate the interference due to ${\bf x}_u$, i.e.,  
we need to download data of size equal to 
\begin{equation} \label{Eqn:repairrank}
\gamma_u=\rank\left(\left[
\begin{array}{c}
\left({\bf R}^{k+1}_i\right)^T{\bf P}^{(k+1)}_{u} \\
\vdots \\
\left({\bf R}^{n}_i \right)^T{\bf P}^{(n)}_u
\end{array}
\right]\right)
\end{equation}
Designing $\mathbf{R}_i^j$ to achieve the following:
\begin{equation}\label{Eqn:repairopt}
\min \sum \limits_{u,u \neq i} \gamma_u \mathrm{~subject~ to~} \gamma_i = \frac{M}{k}
\end{equation}
 is the \textit{repair vector (matrix) design problem}. $\gamma_i = \frac{M}{k}$ means that the useful data matrix must have full rank.

The cut-set bound of \cite{DimakisGWWR:08} states that $\beta$ equations from each of the remaining systematic nodes is the minimum one could achieve, i.e., the minimum rank of each interference space is $\beta$.
This results in a minimum download bound of $\frac{n-1}{n-k}\frac{M}{k}=(n-1)\beta$.
Observe that the above benefits can only be unlocked if we treat each stored symbol as a block of smaller $(n-k)\beta$ sub-symbols.

\subsection{Scalar MDS Codes}
When we consider scalar $(n,k)$-MDS codes, we assume that $k$ information symbols $\mathbf{x}=[x_1\ldots x_k]\in \left(\mathbb{F}_{p^m} \right)^{k\times 1}$ are used to generate $n$ coded symbols $\mathbf{y}=[y_1\ldots y_n] \in \left(\mathbb{F}_{p^m}\right)^{n\times 1}$ under the linear generator map
\begin{equation}\label{Eqn:scalarMDSeq}
\mathbf{y}=
\left[
\begin{array}{c}
   y_1\\
  \vdots \\
   \vdots \\
   \vdots \\
  y_n
\end{array}
\right]=
\left[
\begin{array}{ccc}
1 & \ldots & 0\\
    0       & \ddots & 0 \\
     0      & \ldots  & 1 \\
\hline
P^{(k+1)}_1 & \ldots & P^{(k+1)}_k\\
\vdots & \vdots&\vdots\\
P^{(n)}_k & \ldots & P^{(n)}_k
\end{array}
\right]
\left[
\begin{array}{c}
  x_1 \\ 
 \vdots \\
  x_k
\end{array}
\right]
\end{equation}
where $\mathbf{P}$ is the $k\times (n-k)$ matrix that generates the parity symbols of the code and $\mathbb{F}_{p^m}\equiv\mathbb{GF}(p^m)$. 
Similar to the previous section, let $P^{(j)}_{i}$ denote the parity coefficient, drawn from $\mathbb{GF}(p^m)$, used by the $j$th parity node to multiply symbol $x_i$.
Instead of matrices and vectors in the the previous case, here we have scalars drawn from the extension field $\mathbb{GF}(p^m)$. 
The MDS property is equivalent to the requirement that the $k$ information symbols can be reconstructed from any subset of size $k$.

When a node, or a coded symbol is lost, if we wish to repair it using linear methods over the extension field $\mathbb{GF}(p^m)$, we can perform naive repair. Scalar-linear operations on this code binds us to this worst case repair bandwidth cost. 

Moving away from scalar-linear methods, we could instead download ``parts'' of each symbol defined over $\mathbb{GF}(p^m)$.
Observe that, over $\mathbb{GF}(p^m)$, each symbol consists of $m$ sub-symbols defined over $\mathbb{GF}(p)$ and
$\mathbb{GF}(p^m)$ is isomorphic to  a vector space of dimension $m$ over $\mathbb{GF}(p)$.


In the following section, we describe how an extension field can be used to allow decomposition of each coded symbol into sub-symbols, such that a scalar linear MDS codes is interpreted as a a vector-linear MDS code.
The key ideas used are the following:
\begin{enumerate}
\item  Each element of the generator matrix is viewed as a square matrix with dimensions $m \times m$ over the field $\mathbb{GF}(p)$.
\item Every data symbol $x_i$ and every coded symbol $y_i$ over $\mathbb{GF}(p^m)$ are viewed as vectors $\mathbf{x}_i$ and
$\mathbf{y}_i$, respectively, of dimension $m$ over the field $\mathbb{GF}(p)$. 
\end{enumerate}

\section{Vectorizing scalar codes}\label{Sec:Rep} 

We review some results \cite{lidl1996finite}\cite{macwilliams2006theory} regarding representations of finite field elements.
Let the irreducible primitive polynomial $P(x)$ of degree $m$ over the base field $\mathbb{GF}(p)$ that generates $\mathbb{GF}\left(p^m\right)$ be:
\begin{equation}
\label{Eqn:Irredpoly}
P(x)= a_0 + a_1 x+\ldots a_{m-1} x^{m-1} +x^m,
\end{equation}
where $a_0,\ldots, a_{m-1}\in\mathbb{GF}(p)$.
Let $\zeta$ be any root of the polynomial $P(x)$.  Hence, $\zeta$ is a primitive element. There may be more than one root of the primitive polynomial. All primitive elements are isomorphic to each other (extension fields obtained by setting one of the roots to be the primitive element is isomorphic to the one obtained through other roots)
Then, any field element $b\in \mathbb{GF}(p^m)$ can be written as a polynomial of $\zeta$ over $\mathbb{GF}(p)$ of degree at most $m-1$
\begin{equation}
b= b_0+b_1 \zeta+\ldots b_{m-1}\zeta^{m-1}
\end{equation}
where $b_i \in \mathbb{GF}(p)$, $i\in \{0,1 \ldots m-1 \}$. 

\begin{definition}
The \textit{companion matrix} of the primitive polynomial $P(x)= a_0 + a_1 x+\ldots a_{m-1} x^{m-1} +x^m$ is a $m \times m$ matrix given by:
\begin{align*}
 \mathbf{C}= \left[
 \begin{array}{cccccc}
 0 & 0 & 0 & \cdots & 0 & -a_0 \\
 1 & 0 & 0 & \cdots & 0 & -a_1 \\
 0 & 1 & 0 & \cdots & 0 & -a_2 \\
 0 & 0 & 1 & \cdots & 0 & -a_3 \\
 . &  .  &  .  & \cdots & .  &   .    \\
 . &  .  &  .  & \cdots & .  &   .    \\
  . &  .  &  .  & \cdots & .  &   .    \\
 0 & 0 & 0 & \cdots &  1 & -a_{m-1}
 \end{array}
 \right]
\end{align*}
\end{definition}

\begin{enumerate}
\item \textit{Vector Representation}: Any, $b\in\mathbb{GF}(p^m)$ can be interpreted as a vector that belongs
to a vector space of dimension $m$ over $\mathbb{GF}(p)$ with the following vector representation
\begin{equation}
\label{Eqn:Vectorrep}
f(b) = \left[b_0~b_1~ \ldots ~b_{m-1}\right]^T.
\end{equation}

\item \textit{Matrix Representation}: Any nonzero field element in $\mathbb{GF}(p^m)$ can be written as $\zeta^{n}, ~ 0\leq n \leq p^m-2$. The mapping $g(\zeta^\ell)= \mathbf{C}^{\ell}$ is an isomorphism between $\mathbb{GF}\left( p^m \right)$ and  the set of $m \times m$ matrices $\{0, \mathbf{C}^0,\mathbf{C}^1, \ldots, \mathbf{C}^{p^m-2} \}$ over $\mathbb{GF}(p)$ that preserves the field multiplication and addition in terms of matrix multiplication and addition over the space of matrices $(\mathbb{GF}(p))^{m \times m}$. 
\end{enumerate}

We refer to $g(b)=\mathbf{B}$ as the ``multiplication operator" corresponding to $b \in \mathbb{GF}(p^m)$ and to $f(b)=\mathbf{b}$ as the vector representation of $b \in \mathbb{GF}(p^m)$. Let ${\cal M}\left(\mathbb{F}_{p^m}\right)= \{\mathbf{0},\mathbf{C}^0,\mathbf{C}^1\ldots \mathbf{C}^{p^m-2} \} $ be the set of multiplication operators. Then clearly,
\begin{enumerate}
\item[{\it P1}] \textit{Additivity:}  For any  $c,d \in \mathbb{GF}(p)$ and ${\bf A}, {\bf B} \in {\cal M}\left(\mathbb{F}_{p^m}\right)$, we have
$c{\bf A}+d{\bf B} \in {\cal M}(\mathbb{F})$. 
\item[{\it P2}] \textit{Commutativity:}  For any 
${\bf A},{\bf B} \in {\cal
M}\left(\mathbb{F}_{p^m}\right)$, we have ${\bf A}{\bf B} ={\bf B}{\bf A} \in {\cal
M}\left(\mathbb{F}_{p^m}\right) $.
\end{enumerate}

\begin{lemma}\label{lem:multop}\cite{macwilliams2006theory} 
If $c=ab$ where $c,a,b \in \mathbb{GF}(p^m)$, then $\mathbf{c}=\mathbf{A}\mathbf{b}$ where $f(c)=\mathbf{c},~f(b)=\mathbf{b}$ and $g(a)= \mathbf{A}$.
\end{lemma} 

\subsection{Vectorization of the code in (\ref{Eqn:scalarMDSeq})}

\begin{enumerate}
  \item The information symbols $x_i$ and the coded symbols $y_i$ can be rewritten as $m$-dimensional vectors, ${\bf x}_i$ and ${\bf y}_i$ over $(\mathbb{GF}(p))^{m\times1}$ by setting $\mathbf{x}_i=f(x_i)$ and $\mathbf{y}_i=f(y_i)$. 
   \item Every entry of the generator matrix, i.e. $P_{i}^{(j)}$, can be represented in terms of the multiplication operator $\mathbf{P}_i^{(j)} \in {\cal M}\left(\mathbb{F}_{p^m}\right)$ by setting $\mathbf{P}_i^{(j)}=g(P_i^{(j)})$.
   \item By Lemma \ref{lem:multop}, every $P_i^{(j)} x_i$ is represented by the matrix-vector multiplication $\mathbf{P}_i^{(j)} \mathbf{x}_i$.
\end{enumerate}

 Setting $m=(n-k)\beta$, we observe that we have changed the scalar code in (\ref{Eqn:scalarMDSeq}) into the vector code given by (\ref{Eqn:vectorMDSeq}). The reason for the choice of $m$ has been given in \ref{Sec:VectorMDS}. The only difference between this and a generic vector code is that the matrices $\mathbf{P}_{i}^{(j)}$ are multiplication operators that have specific structure. Note that the same construction has been recently used in \cite{hong2013structured}. Also, this construction can be taken to be the finite field analogue of the procedure for generating irrational dimensions out of a real dimension \cite{motahari2009real} that plays an important role in \textit{real interference alignment}.

\section{Repair Field Elements}
	At this point, one could consider the vectorized code obtained to be a generic vector linear code and design repair matrices $\mathbf{R}_i^j$  to solve (\ref{Eqn:repairopt}) by searching over all possible repair matrices. Any such design seems to depend on the structure of the multiplication operators $\mathbf{P}_i^{(j)} $. However, we use the following technical lemma to illustrate that designing repair matrices (or repair vectors) as in (\ref{Eqn:repairopt}) can be cast as a problem of designing \textit{repair field elements}, when it comes to repairing a vectorized scalar code. This lets us bypass the need for "looking into" the structure of the multiplication operators and the need for checking all possible repair matrices. This is the main technical idea behind the paper.

\vspace{0.3cm}
\begin{lemma} \label{Lemma:multexist}
	For any two nonzero vectors $\mathbf{a}^T,\mathbf{b}^T \in \mathbb{GF}(p)^{1 \times m}$,
there always exists a multiplication
operator (or a matrix) $\mathbf{M} \in {\cal M}\left(\mathbb{F} \right)$ ($m
\times m$) such that $\mathbf{b}^T
\mathbf{M} = \mathbf{a}^T$.
\end{lemma}
\vspace{0.3cm}
\begin{proof}
The proof is provided in the appendix.
\end{proof}

\textbf{Remark:} We have represented the multiplication of $c=ab$ over the extension field as $\mathbf{c}=\mathbf{A}\mathbf{b}$, where $\mathbf{A}$ is a multiplication operator and $\mathbf{b}$ is the vector representation. But this corresponds to right multiplication \textit{only}. Clearly, $\mathbf{c}^T=\mathbf{b}^T\mathbf{A}$ is not true as the matrix $\mathbf{A}$ in general is not symmetric. Hence, we require Lemma \ref{Lemma:multexist} to establish properties when the matrix is multiplied by a vector from the \textit{left}. 

Consider the repair problem for the vectorized code as in Section \ref{Sec:VectorMDS} and use Lemma \ref{Lemma:multexist}.
We download $\beta$ equations from every node since we have vectorized over the field $\mathbb{GF}(p)$ and $m=(n-k)\beta$.
Without loss of generality, let us consider the repair of node $i=1$. As in Section \ref{Sec:VectorMDS}, repair matrices which multiply the $n-k$ parities are denoted ${\bf R}^{k+1}, \ldots {\bf R}^{n}\in\left(\mathbb{GF}(p)\right)^{\beta(n-k) \times \beta}$, dropping the subscript $i$ since we will state everything for repair of node $1$. Let $ \mathbf{r}^{\ell}_j $ denote the $j$-th column of ${\bf R}^{\ell}$. This corresponds to the $j$th equation downloaded from node $\ell$.

Now, due to Lemma \ref{Lemma:multexist}, we can fix an arbitrary nonzero $\left( \mathbf{r}\right)^T$ as a reference vector. Then
$\forall \ell,~\forall j,~ \exists \mathbf{M}^{\ell}_j \in {\cal M}(\mathbb {F}_{p^m})$,
\begin{equation} \label{Eqn:repairwrite}
	\left(\mathbf{r}^{\ell}_j \right)^T = \left( \mathbf{r}
\right)^T\mathbf{M}^{\ell}_j .
\end{equation}
 Hence, all the repair vectors can be replaced by the \textit{repair field elements} $M^{\ell}_j$ corresponding to the operators $\mathbf{M}^{\ell}_j$ .  In terms of the repair field elements, we have the following important theorem that gives an alternate characterization of any repair scheme.

\vspace{0.3cm}
\begin{theorem}\label{Thm:RepairField}
Consider a repair scheme with repair matrices $\mathbf{R}^{\ell},~k+1 \leq \ell \leq n$ for repairing node $1$ (for any node $i$ in general) where the rank of the column of vectors corresponding to the $u$th data vector as in (\ref{Eqn:repairrank}) is
\begin{equation} 
\gamma_u=\rank\left(\left[
\begin{array}{c}
\left({\bf R}^{k+1} \right)^T{\bf P}^{(k+1)}_{u} \\
\vdots \\
\left({\bf R}^{n} \right)^T{\bf P}^{(n)}_u
\end{array}\nonumber
\right]\right).
\end{equation}

This is possible if and only if one can find repair field elements $M^{\ell}_j \in \mathbb{F}_{p^m}$ for every equation $j \in[1,\beta]$ downloaded from every node $ \ell \in[k+1,n]$, such that $\forall u$:  
{\small
\begin{align}\label{Eqn:rankcond}
\gamma_u &= \rank_p \left(M^{k+1}_1 P_u^{k+1}, \dots M^{k+1}_{\beta} P_u^{k+1}, \dots
M^{n}_1 P_u^{(n)}, \right.  \nonumber \\ 
& \qquad \left. \dots ,M^{n}_{\beta}P_u^{(n)} \right)
\end{align}
}where $\rank_p(a_1,\dots a_m)$ with $a_i\in\mathbb{GF}(p^m)$ is defined according to the following restricted definition of linear independence:
A set of field elements $A_i \in \mathbb{GF}(p^m)$ are linearly independent over a sub-field $\mathbb{GF}(p)$ (i.e., have $\rank_p$ equal to the cardinality of this set of elements), when one cannot find non zero scalars
$v_i \in \mathbb{GF}(p)$ such that $\sum v_i A_i =0$.
\end{theorem} 
\begin{proof}
The proof uses Lemma 3, and is relegated to the appendix.
\end{proof}
\vspace{0.3cm}
We call the above formulation of the repair problem as {\it repair field element design},  as one needs to design one repair field element (i.e $M^{\ell}_j$'s) for every repair vector $\mathbf{r}^{\ell}_j$ or equivalently for every downloaded equation. The repair bandwidth (in bits) of any such scheme with repair field elements is proportional to the sum of the ranks, i.e. $\sum_{i=1}^k r_i \log_2 \left(p \right)$ bits where the code is subpacketized over $\mathbb{GF}(p)$.

\subsection{Illustration of repair of a (5,3) Reed-Solomon Code} \label{Sec:53reed}

Consider a $(5,3)$-Reed Solomon code over $\mathbb{F}=\mathbb{GF}(2^4)$. Let $\omega$ be the fifth root of unity. 
Using the explicit formula for the generator matrix of the systematic Reed Solomon code given in \cite{versfeld2010systematic}, we obtain the following structure for the generator matrix
$\mathbf{G}$
\begin{equation}
	\mathbf{G} = \left[ 
	\begin{array}{ccccc}
	1 & 0 & 0 &
\frac{(\omega^4-\omega^2)(\omega^4-\omega^3)}{(\omega-\omega^2)(\omega-\omega^3)
} &
\frac{(\omega^5-\omega^2)(\omega^5-\omega^3)}{(\omega-\omega^2)(\omega-\omega^3)
} \\
	0 & 1 & 0 &
\frac{(\omega^4-\omega)(\omega^4-\omega^3)}{(\omega^2-\omega)(\omega^2-\omega^3)
} &
\frac{(\omega^5-\omega)(\omega^5-\omega^3)}{(\omega^2-\omega)(\omega^2-\omega^3)
} \\
	0 & 0 & 1 &
\frac{(\omega^4-\omega)(\omega^4-\omega^2)}{(\omega^3-\omega)(\omega^3-\omega^2)
} &
\frac{(\omega^5-\omega)(\omega^5-\omega^2)}{(\omega^3-\omega)(\omega^3-\omega^2)
} 
	\end{array}
	\right].\nonumber
\end{equation}
For the repair problem, without loss of generality, the generator matrix given above can be simplified by factoring out some coefficients along every row and renormalizing so that we can work on the following equivalent generator matrix
	\begin{equation}
	\mathbf{G} = \left[ 
	\begin{array}{ccccc}
		1 & 0 & 0 & 1 & \omega^2+\omega+1 \\
		0 & 1 & 0 & 1 & \omega \\
		0 & 0 & 1 & 1 & \omega^2+1 
	\end{array}
	\right].
	\end{equation}
Let $\zeta$ be the primitive element of $\mathbb{F}$ corresponding to the primitive polynomial $P(x)=1+x+x^4$. Then, $\zeta^{15}=1$ and $\omega =\zeta^3$. 

Now, we consider the repair problem under the vector representation of the code over $\mathbb{GF}(2)$. 
For this case, $\beta=2$. 
Hence, $\beta=2$ equations are downloaded from the $2$ parity nodes.  
The cut-set bound (from the optimal repair bound) for this scenario is downloading $8$ equations in total.

\begin{table}
\centering
\begin{tabular}{|c|p{5cm}|}
\hline 
\multirow{4}{*}{Repair for node $1$} & $g_1(\omega)=\omega^3+1  $\\
& $g_2(\omega)=\omega^2+1  $\\
& $f_1(\omega)=\omega g_1(\omega)$ \\
& $f_2(\omega) =g_2(\omega)(\omega^2+1) $ \\ \hline
\multirow{4}{*}{Repair for node $2$} & $g_1(\omega)=\omega+1 $\\ 
& $g_2(\omega)=\omega$ \\
& $f_1(\omega)=(\omega^2+\omega+ 1)g_1(\omega)$\\ 
& $f_2(\omega) =g_2(\omega)(\omega^2+1) $  \\
\hline 
\multirow{4}{*}{Repair for node $3$} & $g_1(\omega)=\omega+1 $ \\
& $ g_2(\omega)=\omega $ \\ 
& $f_1(\omega)=(\omega^2+\omega+ 1)g_1(\omega) $ \\
& $ f_2(\omega) =g_2(\omega)\omega $ \\              
\hline
\end{tabular}
\caption {Repair field elements for the repair of systematic nodes for the (5,3) RS code.}
\label{Tab:Repair5,3}
\end{table}

The polynomials in Table \ref{Tab:Repair5,3} correspond to
repair by downloading $10$ equations in the event of a failure of any systematic node for the $(5,3)$ Reed Solomon code
over $\mathbb{GF}(16)$.  Now, we illustrate this using the framework of repair field elements. Each parity node stores four equations over the binary field. Let us consider the repair of node $1$. $f_i(\omega)$ denotes the repair field element corresponding to the bit $i$ downloaded from the first parity node (node $4$) and $g_i(\omega)$ denotes the repair field element for the second parity node (node $5$). Let us assume that the corresponding repair vectors are $\left(\mathbf{r}^4_1\right)^T, \left(\mathbf{r}^4_2 \right)^T,\left(\mathbf{r}^5_1\right)^T$ and $ \left(\mathbf{r}^5_2 \right)^T$.

As an illustration of the results in this section, we show how the repair field elements in Table \ref{Tab:Repair5,3} correspond to a repair bandwidth of $10$ equations over $\mathbb{GF}(2)$ for repair of node $1$. First, we show how the column of vectors corresponding to data node $1$ in (\ref{Eqn:repaireq}) is full rank, if we use the repair vectors obtained through the repair field elements in  Table \ref{Tab:Repair5,3}. We need to verify the following:
\begin{align*}
 \gamma_1 & = \rank_2 \left(f_1(\omega),f_2(\omega),g_1(\omega)(\omega^2+\omega+1), \right.\\
 & \qquad  \left. g_2(\omega)(\omega^2+\omega+1) \right) = 4.
\end{align*}

 By expressing everything as a polynomial in $\zeta$ of degree at most $3$ with coefficients from $\mathbb{GF}(2)$ using the irreducible polynomial, we have
 \begin{align}\label{Eqn:fullrank53}
 \left[
 \begin{array}{c}
 f_1(\omega)\\
 f_2(\omega)\\
 g_1(\omega)(\omega^2+\omega+1)\\
 g_2(\omega)(\omega^2+\omega+1) 
 \end{array}
  \right]
 = 
  \left[ 
  \begin{array}{cccc}
  0 & 1 & 1 & 1 \\
  1&1&1&0\\ 
  0&0&0&1 \\ 
  1&1&0&0 
  \end{array}
  \right] \left[
  \begin{array}{c}
  \zeta^3 \\ 
  \zeta^2 \\
   \zeta \\
    1
    \end{array}
  \right].
 \end{align}
  We see that they are linearly independent over $\mathbb{GF}(2)$ (full rank). Hence, $\gamma_1 =4$. Now, we set $\left(\mathbf{r}^4_1 \right)^T=\left(\mathbf{r}\right)^T f_1(\mathbf{W})$, $\left(\mathbf{r}^4_2 \right)^T=\left(\mathbf{r}\right)^T f_2(\mathbf{W})$, $\left(\mathbf{r}^5_1 \right)^T=\left(\mathbf{r}\right)^T g_1(\mathbf{W})$ and $\left(\mathbf{r}^5_2 \right)^T=\left(\mathbf{r}\right)^T g_2(\mathbf{W})$
where $\mathbf{W}$ is the multiplication operator for $\omega$.  By applying Theorem \ref{Thm:RepairField}, the column of vectors corresponding to data from node $1$ as in (\ref{Eqn:repaireq}) given by:
\begin{align*}
 \left[
\begin{array}{c}
\left({\bf r}^{4}_1\right)^T \left({\bf P}^{(4)}_{1}\right)^T \\
 \left({\bf r}^{4}_2\right)^T \left({\bf P}^{(4)}_{1}\right)^T \\
\left({\bf r}^{5}_1\right)^T \left({\bf P}^{(4)}_{1}\right)^T \\
\left({\bf r}^{5}_2\right)^T \left({\bf P}^{(4)}_{1}\right)^T 
\end{array}
\right]=    
 \left[
\begin{array}{c}
 \mathbf{r}^T f_1(\mathbf{W})\\
\mathbf{r}^T f_2(\mathbf{W}) \\ 
\mathbf{r}^Tg_1(\mathbf{W}) (\mathbf{W}^2+\mathbf{W}+\mathbf{I}) \\
\mathbf{r}^Tg_2(\mathbf{W}) (\mathbf{W}^2+\mathbf{W}+\mathbf{I}) \\
\end{array}
\right] 
\end{align*}
 is full rank over $\mathbb{GF}(2)$, with the above assignment of repair vectors because of (\ref{Eqn:fullrank53}). Here, $\mathbf{r}$ is any arbitrary non-zero reference vector.

Now, the rank of interference terms follows from similar observations regarding repair field elements: $\omega g_1(\omega)= f_1(\omega)$ implies $\gamma_2=3$ (making column for data vector $3$ to have rank $3$) and $(\omega^2+1) g_2(\omega)= f_2(\omega)$ implies $\gamma_3=3$. Therefore, $\sum \gamma_i=10$ equations over $\mathbb{GF}(2)$ needs to be downloaded for repair for node 1.  Repair bandwidth for repair of nodes $2$ and $3$ can be verified similarly.

 Now, We argue that $10$ bits is the optimal linear repair bandwidth achievable for this code.
 We consider the case where node 2 fails and we will assume that $8$ repair equations are sufficient.
To recover the lost data, according to Eq. (\ref{Eqn:rankcond}), we require 
\begin{equation}\label{Eqn:fullrank}
\rank_2\left[f_1(\zeta)~ f_2(\zeta)~ g_1(\zeta)\zeta^3~  g_2(\zeta)\zeta^3
\right] = 4.
\end{equation}
  
  If $8$ equations are sufficient then there must exist polynomials such that the following conditions are
true:
      \begin{eqnarray}
        \rank \left[f_1(\zeta)~ f_2(\zeta)~ g_1(\zeta)(\zeta^6+\zeta^3+1)
\right. \nonumber\\
   \left. g_2(\zeta)(\zeta^6+\zeta^3+1)) \right] = 2, \nonumber \\
        \rank \left[f_1(\zeta)~ f_2(\zeta)~ g_1(\zeta)(\zeta^6+1)~ 
g_2(\zeta)(\zeta^6+1)) \right] 
 =  2 .
      \end{eqnarray}
      
        Then the only possibility is that $g_1(\zeta)(\zeta^6+\zeta^3+1) =
v_{1} f_1(\zeta+v_{2}
 f_2(\zeta)$ and $g_2(\zeta)(\zeta^6+\zeta^3+1) = v_{3}
f_1(\zeta)+v_{4} f_2(\zeta)$. Similarly,
  $g_1(\zeta)(\zeta^6+1) = v_{5} f_1(\zeta)+v_{6} f_2(\zeta)$ and
$g_2(\zeta)(\zeta^6+1)
 =v_{7} f_1(\zeta)+v_{8} f_2(\zeta)$. Here, all $v_i \in \mathbb{GF}(2)$.
 
 Therefore,  
       \begin{align}
           g_1(\zeta)(\zeta^3) & =& (v_{1}+v_{5})
f_1(\zeta)+(v_{2}+v_{6}) f_2(\zeta), \nonumber \\
          g_2(\zeta)(\zeta^3) & =& (v_{3}+v_{7})
f_1(\zeta)+(v_{4}+v_{8}) f_2(\zeta). \nonumber
       \end{align}
 
This violates the full rank condition of (\ref{Eqn:fullrank}). Similar
arguments hold for repair of systematic nodes $1$ and $3$. Further, very similar arguments can be made to show that $9$ equations are not enough for repair of the nodes. The arguments are lengthy but follow a similar style to the one above. The crucial property that is used in these converse results is the following property: In the fifth row $[1+\zeta^6+\zeta^3,\zeta^3,\zeta^6+1]$  of the generator matrix, two coefficients add up to give the third coefficient.

\subsection{Different degrees of subpacketization}
Note that we made no assumption about $\mathbb{GF}(p)$. So this could be an extension field by itself.  So, for a given extension field $\mathbb{GF}(p^{(n-k)r})$, where $p$ is prime,
one could do the vectorization over $\mathbb{GF}(p^{r})$, so that the effective degree of subpacketisation is $\beta=1$. The lowest degree is $1$ and the highest possible degree is $r$ and any intermediate degree would be any $s$ that divides $r$. The following intuitive result shows that any repair scheme for a lower degree of subpacketization $\beta'$ can be 
 implemented using an equivalent repair scheme with a higher degree of subpacketization $r\beta'$ and with the same benefits with respect to the repair
 bandwidth.
 
 \vspace{0.3cm}
 \begin{lemma}\label{Thm:HighLowSub}
 Consider a scalar systematic $(n,k)$ MDS code over a field $\mathbb{GF}(p^{a(n-k)})$, vectorized over $\mathbb{GF}(p^a)$, with $\beta=1$.
 Let the associated repair field elements be $M^{k+1},M^{k+2}\ldots M^{n}$ that operate on the parity nodes numbering from $k+1$ to $n$. Consider the rank of the column of vectors in (\ref{Eqn:rankcond}) corresponding to node $i$ given by: $r_i = \rank_{p^a} \left( \left[M^{k+1}
 P^{k+1}_i ~ M^{k+2}P^{k+2}_i \ldots M^{n}P^{n}_i\right] \right)$. Define new repair field elements, for the code subpacketized over $\mathbb{GF}(p)$ to be: $\left[\tilde{M}_1^{i}~\tilde{M}_2^{i} \ldots \tilde{M}_{a}^{i} \right]=\left[M^{i},~M^{i}\zeta,\dots M^{i}\zeta^{a-1}\right]$ corresponding to $a$ equations being drawn from node $i$, where $k+1 \leq i \leq n$. Here, $\zeta \in \mathbb{GF}(p^a)$ is the primitive element. The system of new repair elements $\{\tilde{M}_j^i \}$ have the same repair bandwidth as $\{M_i \}$.
 \end{lemma}
 \vspace{0.3cm}
 \begin{proof}
  The proof is relegated to the appendix.
 \end{proof}

\section{Clique Repair}\label{Sec:cliquerepair}
In this section, we use the repair field elements framework to prove the following theorem that gives an optimal repair scheme when $\beta=1$ for any $(n,n-2)$ scalar MDS code. We call this scheme \textit{Clique Repair}. From now on, without loss of generality, we assume that $P_{j}^{(k+1)} =1$, i.e. all parity coefficients for parity node $k+1$ are $1$. This can be justified as it does not affect the MDS repair problem.
\vspace{0.3cm}
\begin{theorem}\label{Thm:WuDimakis}
Consider a systematic $\left(n,n-2 \right)$-MDS code over $\mathbb{GF}\left(p^{2r}\right)$ and an undirected graph $G(V,E)$ such that $|V| = k$ and $(i,j) \in E$ iff $P^{(k+2)}_i \left(P^{(k+2)}_j\right)^{-1} \in \mathbb{GF}\left(p^r \right)$. 
Then, with linear repair schemes, node $i$ cannot be repaired with BW less than $M-\frac{C_i}{2}\frac{M}{k}$ when vectorized over $\mathbb{GF}\left(p^{r}\right)$, where $C_i$ is the size of the largest clique of $G$ not containing node $i$. 
\end{theorem}
\begin{proof}
The proof is relegated to the appendix.
\end{proof}





\textbf{Remark:} There is an alternative way to see the above theorem. $\mathbb{GF}(p^r) \backslash \{0\}$ is a multiplicative subgroup of $\mathbb{GF}(p^{2r}) \backslash \{0\}$. Consider the set of cosets formed by the subgroup  $\mathbb{GF}(p^r) \backslash \{0\}$. Consider the repair of node $i$. Among all cosets that do no contain the field element $P^{k+2}_i$,  pick the coset that contains the largest number of elements from $\{P^{k+2}_j\}_{j \neq i}$. Let the number of elements from $\{P^{k+2}_j\}_{j \neq i}$ which lie in this coset be $C_i$. Then, the repair bandwidth is no less than $M - \frac{C_i}{2} \frac{M}{k}$ in terms of $\mathbb{GF}(p^r)$ symbols.

Although the theorem above only specifies a lower bound, one can come up with an algorithm to achieve the optimum performance. It is easy to check that the following algorithm works. 
\begin{algorithm}
     \caption{\it{Generate Clique}}
     \begin{algorithmic}
      \WHILE{$i=1 \to k$}
          \WHILE {$j =1 \to i-1$ } 
            \IF {$P^{(k+2)}_i\left(P^{(k+2)}_j\right)^{-1} \in \mathbb{GF}(p^r)$}
                 \STATE $E \leftarrow (i,j)$
             \ENDIF    
           \ENDWHILE       
      \ENDWHILE
     \end{algorithmic}
   \end{algorithm}
The algorithm Generate Clique identifies the disjoint cliques (or cosets). Let us assume that there is a list $\{C[i]\}_{1 \leq i \leq m}$  such that $C[i]$ contains all vertices contained in clique $i$ or coset $i$. The algorithm Find Repair finds the optimal repair field element $\mu$ for repairing node $i$. 
    \begin{algorithm}
     \caption{\it{Find Repair}}
     \begin{algorithmic}
     \STATE $\rm{Find~}$ $N: i \in C[N]$
     \STATE $k_{\rm{max}} \leftarrow \arg \max \limits_{k \neq N} \lVert C[k] \rVert $   
      \STATE $\rm{Pick~ some ~ node~}\ell \in C_{k_{\rm{max}}}$.
      \STATE $\mu \leftarrow \left(P^{(k+2)}_{\ell}\right)^{-1}$
     \end{algorithmic}
   \end{algorithm}
Notice that the algorithm runs using $ O(n^2)$ field multiplication operations. The scheme gives an analytical connection between the repair BW and the coefficients of the generator matrix (see remark after Theorem \ref{Thm:WuDimakis}).

Consider the vector representation of the $(5,3)$ RS code in Section \ref{Sec:53reed} over $\mathbb{GF}(2^2)$.  Then by applying 
Theorem \ref{Thm:WuDimakis}, we find that all the three nodes lie in the same clique. In other words, $(\omega^2+\omega+1)^{-1} \omega$, $\omega (\omega^2+1)^{-1}$ belong to
$\mathbb{GF}(2^2)$. Hence, for this code, clique repair does not give any gain in terms of repair bandwidth.

 Now, we present examples of bandwidth savings that are possible for a $(6,4)$ Reed Solomon code and for the $(14,10)$ Reed Solomon code employed in HDFS open source module. As we will see, clique repair gives nontrivial bandwidth savings over naive repair for the $(6,4)$ Reed Solomon code considered below. This can be improved further by going to a higher degree of subpacketization.

\section{Analysis of repair of (6,4) Reed Solomon Codes}

Here, we consider a $(6,4)$-RS code over $\mathbb{GF}(2^4)$. Let $\zeta$ be the primitive element of $\mathbb{F}$ corresponding to the primitive polynomial $P(x)=1+x+x^4$
Using the formula in
\cite{versfeld2010systematic}, we obtain the following systematic generator matrix
\begin{equation}
\mathbf{G}= \left[ 
	\begin{array}{cccccc}
	1 & 0 & 0 &  0 &1 & \zeta^3+\zeta^2+\zeta+1 \\
	0 & 1 & 0 &  0 &1 &  \zeta+1\\
	0 & 0 & 1 &  0 &1 &  1\\
	0 & 0 & 0 &  1 &1 &  \zeta^2             
	\end{array}
	\right]
\end{equation}
We consider the vector representation of the code over $\mathbb{GF}(2^2)$ ($\beta=1$). 
If we apply Theorem \ref{Thm:WuDimakis} to this code, there are $3$ cliques that are formed. 
The first clique (or coset) contains nodes $1$ and $4$ while the second one contains $2$ and the third one contains node $3$. 
By the clique repair algorithm presented in Section \ref{Sec:cliquerepair}, the repair of nodes $2$ and $3$ require $6$ repair equations over $\mathbb{GF}(4)$ to be downloaded. For the repair of nodes $1$ and $4$, $7$ equations over $\mathbb{GF}(4)$ need to be downloaded which is close to the file size $M=8$, while the cut-set bound is $\frac{n-1}{n-k}\frac{M}{k}=5$
equations.  

Now, consider a higher degree of subpacketization, i.e. each node stores $4$ elements over $\mathbb{GF}(2)$. Now, $M=16$ elements. We get a good repair scheme (by Lemma \ref{Thm:HighLowSub}) for nodes $2$ and $3$ over   $\mathbb{GF}(2)$ that requires $12$ equations by converting the clique repair scheme for these nodes over $\mathbb{GF}(2^2)$. Hence, the repair bandwidth for $2$ and $3$ is $6 \times 2=12$ equations for repair over $\mathbb{GF}(2)$ . The cut set bound is $5 \times 2 =10$  equations over $\mathbb{GF}(2)$.

For this case, the repair scheme with repair field elements given in Table \ref{Tab:Repair6,4} improves the repair BW for nodes $1$ and $4$ to $12$ equations compared to the clique repair equivalent that requires $7 \times 2=14$ equations. $\{f_i\}$ represent repair field elements for the first parity node and $\{g_i\}$ represent repair field elements for the second parity node. It is possible to show that $12$ equations is the optimal linear repair bandwidth for this code. The argument is lengthy, but similar in style to the one in Section \ref{Sec:53reed} for the (5,3) Reed-Solomon code and hence we skip it.

\begin{table}
\centering
\begin{tabular}{|c|p{5cm}|}
\hline 
\multirow{4}{*}{Repair for node $1$} & $g_1(\zeta)=\zeta^{-2}  $\\
& $g_2(\zeta)=1$\\
& $f_1(\zeta)=1$ \\
& $f_2(\zeta) =\zeta^2$ \\ \hline
\multirow{4}{*}{Repair for node $4$} & $g_1(\zeta)=1$\\ 
& $g_2(\zeta)=\zeta$ \\
& $f_1(\zeta)=1$\\ 
& $f_2(\zeta)=\zeta$  \\
\hline 
\end{tabular}
\caption {Repair field elements for repair of nodes $1$ and $4$ for the systematic (6,4) RS code with subpacketization over $\mathbb{GF}(2)$.}
\label{Tab:Repair6,4}
\end{table}

\section{Numerical Results on the $(14,10)$ Reed-Solomon Code implemented in the Hadoop file system}

The Apache Hadoop Distributed File System (HDFS) relies by default on block replication for data reliability. 
A module called HDFS RAID (\cite{sathiamoorthy2013xoring, raidwiki}) was recently developed for HDFS that allows the deployment of Reed-Solomon and also more sophisticated distributed storage codes. HDFS RAID is currently used in production clusters including Facebook analytics clusters storing more than $30$ PB of data. In this section, we present numerical results on improving the repair performance of the specific $(14,10)$ Reed-Solomon code implemented in HDFS-RAID~\cite{raidwiki}.  

HDFS RAID implements a systematic Reed Solomon code over the extension field $\mathbb{GF}(2^8)$. Let $\zeta$ be the root of the primitive polynomial $1+x^2+x^3+x^4+x^8$ that generates the extension field. The generator matrix used is:
    \begin{align*}
        \mathbf{G}= \left[
        \begin{array}{c}
        \mathbf{I}_{10} \\ 
       \mathbf{P}
        \end{array}
         \right]
    \end{align*}
  where $\mathbf{P}$ is a $4 \times 10$ matrix given by:
  \begin{align*}
      \mathbf{P}^T= \left[ 
       \begin{array}{cccc}
        \zeta^6 &   \zeta^{78} &  \zeta^{249} &   \zeta^{75} \\
       \zeta^{81} &  \zeta^{59} &  \zeta^{189}  &  \zeta^{163} \\
      \zeta^{169} & \zeta^{162} & \zeta^{198} & \zeta^{131}\\
   \zeta^{137}& \zeta^{253}  & \zeta^{49} & \zeta^{143}\\
   \zeta^{149} & \zeta^{177} & \zeta^{96} & \zeta^{205}\\
   \zeta^{211} &  \zeta^{71}  & \zeta^{157} & \zeta^{134}\\
   \zeta^{140} & \zeta^{236} & \zeta^{154} &  \zeta^{43}\\
    \zeta^{49}  & \zeta^{213} &  \zeta^{112} & \zeta^{ 88}\\
    \zeta^{94}  & \zeta^{171} & \zeta^{138} & \zeta^{ 95}\\
   \zeta^{101} &  \zeta^{13} &  \zeta^{148} &  \zeta^{173}
      \end{array}
      \right]
  \end{align*}
    Since the number of parities is $4$ ($n-k=4$), the clique repair technique is not applicable. We consider repair with the highest possible subpacketization, i.e. $\beta=2$ and each node stores $(n-k)\beta = 8$ elements over $GF(2)$ and $M=80$. The repair requires downloading $2$ equations form every parity node. We provide a repair scheme in terms of the eight repair field elements $M_1^{11},M_2^{11}, \ldots M_1^{14}, M_2^{14}$, belonging to $\mathbb{GF}(2^8)$, as in Theorem \ref{Thm:RepairField}. The repair scheme, given in Table \ref{repair1014}, lists the repair field elements for repair of each node and the total number of equations to be downloaded for repair in each case. The average number of equations to be downloaded is  $64.2$ equations. The naive repair involves downloading $80$ equations and the lower bound $\frac{n-1}{n-k}\frac{M}{k}$ gives $26$ equations. We note that the repair scheme that we provide is not the optimal for the code because an exhaustive search involves checking a huge number of combinations (about $2^{64}$ combinations) of the repair field elements. We have searched over about $100000$ random combinations of the repair field elements to produce this repair scheme that saves about $20$ percent bandwidth over naive repair.
    
    \textbf{Remark:} In \cite{sathiamoorthy2013xoring}, a new implementation of a locally repairable code based on the $(14,10)$ code is used to optimize repair. It saves $50$ percent bandwidth over naive repair but incurs a cost of $14$ percent in additional storage overhead. We have demonstrated that the $(14,10)$ code used "as is" without any storage overhead can give non-trivial savings.
    \begin{table*}
    \centering
    \caption{Repair Scheme for $(14,10)$ Reed Solomon Code that saves $20$ percent in repair bandwidth. The naive repair involves downloading $80$ bits for repair.}
    \begin{tabular}{|c | c | c |}
        \hline
        Systematic node repaired & Repair field elements $\left[M_1^{11}~ M_2^{11} \ldots M_1^{14}~M_2^{14} \right]$ &Repair Bandwidth (bits downloaded) \\    
        \hline
        1                                             & $\left[\zeta^{69}~\zeta^{203}~\zeta^{189}~\zeta^{64}~\zeta^{170}~\zeta^{173}~\zeta^{64}~\zeta^{174} \right]$ & $65$ \\
       
        2                                             & $\left[\zeta^{8}~\zeta^{191}~\zeta^{175}~\zeta^{248}~\zeta^{18}~\zeta^{1}~\zeta^{69}~\zeta^{126} \right]$ & $64$ \\
      
        3                                            & $\left[\zeta^{153}~\zeta^{15}~\zeta^{101}~\zeta^{3}~\zeta^{223}~\zeta^{179}~\zeta^{114}~\zeta^{14} \right]$ & $64$ \\
        
        4                                           & $\left[\zeta^{92}~\zeta^{86}~\zeta^{31}~\zeta^{129}~\zeta^{67}~\zeta^{213}~\zeta^{67}~\zeta^{144} \right]$ & $64$ \\
      
        5                                          & $\left[\zeta^{46}~\zeta^{213}~\zeta^{86}~\zeta^{151}~\zeta^{28}~\zeta^{169}~\zeta^{69}~\zeta^{146} \right]$ & $63$ \\
    
           6                                          & $\left[\zeta^{83}~\zeta^{164}~\zeta^{182}~\zeta^{116}~\zeta^{104}~\zeta^{185}~\zeta^{245}~\zeta^{178} \right]$ & $64$ \\
     
           7                                         &  $\left[\zeta^{57}~\zeta^{48}~\zeta^{14}~\zeta^{111}~\zeta^{195}~\zeta^{60}~\zeta^{221}~\zeta^{132} \right]$ & $64$ \\
       
           8                                         &  $\left[\zeta^{51}~\zeta^{174}~\zeta^{206}~\zeta^{224}~\zeta^{104}~\zeta^{100}~\zeta^{52}~\zeta^{143} \right]$ & $65$ \\
       
          9                                         &  $\left[\zeta^{84}~\zeta^{250}~\zeta^{143}~\zeta^{76}~\zeta^{21}~\zeta^{225}~\zeta^{207}~\zeta^{105} \right]$ & $65$ \\
         
         10                                         &  $\left[\zeta^{161}~\zeta^{180}~\zeta^{131}~\zeta^{89}~\zeta^{69}~\zeta^{37}~\zeta^{15}~\zeta^{177} \right]$ & $64$ \\ 
        \hline   
    \end{tabular}  
   \label{repair1014}
    \end{table*}
\section{Conclusion}
       We introduced a framework for repairing scalar codes by treating them as vectors over a smaller field. This is achieved by treating multiplication of scalar field elements in the original field as a matrix-vector multiplication operation over the smaller field. Interference alignment conditions map to designing repair field elements in the large field. Further using the conditions on designing repair field elements, we introduced the \textit{clique repair} scheme for two parities when the degree of subpacketization is $1$, which establishes a connection between the coefficients of the generator matrix and the repair schemes possible. We exhibited good repair schemes for a few Reed-Solomon codes including the one currently deployed in Facebook.
       
       This work hints at the existence of scalar MDS codes with good repair properties. An interesting problem would be to come up with easily testable analytical conditions, similar in spirit to the clique repair scheme, for codes with larger number of parities and for higher degrees of subpacketization.  Sufficient conditions for a specific class of codes like Reed Solomon would be also interesting. More generally, it seems that scalar MDS codes with near optimal repair could be designed using this framework.
 
\section{Acknowledgement}
   We thank the anonymous reviewers for their helpful comments and suggestions that helped us improve the paper immensely.  
\bibliographystyle{IEEEtran}
\bibliography{MDSrepairbib}

\begin{appendix}
\subsection{Proof of Lemma \ref{Lemma:multexist}}
\begin{proof}
		We note that multiplication of the matrix $\mathbf{M}$ from
the left by $\mathbf{a}^T$ does not represent field multiplication. Hence, with respect to left multiplication, the matrix $M$
does not \textit{necessarily} act as a
multiplication operator. The theorem implies that given any two arbitrary non
zero repair vectors, one can find a
multiplication matrix that connects both. 

		Since, non zero field elements in $\mathbb{F}$ are finite,
there are finitely many operators in ${\cal
M}\left(\mathbb{F}\right)$. Let them be denoted by
$\mathbf{M}_1,\mathbf{M}_2,\ldots ..\mathbf{M}_{p^m-1}$. All these
matrices (or operators) have full rank. We consider the products,
$\mathbf{a}^T\mathbf{M}_i$. We show that all of them
are distinct. Suppose for some $i \neq j, \mathbf{a}^T \mathbf{M}_i =
\mathbf{a}^T \mathbf{M}_j $, then 
\begin{equation}
	\mathbf{a}^T \left(\mathbf{M}_i-\mathbf{M}_j\right) = 0.
\end{equation}
But $\mathbf{M}_i -\mathbf{M}_j$ is another multiplication operator by
additivity property. It is non zero and has
full rank since $\mathbf{M}_i \neq \mathbf{M}_j$. This means all the $p^m-1$
products are different. Since there are only
$p^m-1$ non zero repair vectors $\mathbf{b}^T$, given any $\mathbf{b}^T$, one
can always find a $\mathbf{M}_j$ such that
$\mathbf{a}^T \mathbf{M}_j = \mathbf{b}^T$.
\end{proof}

\subsection{Proof of Theorem \ref{Thm:RepairField}}
\begin{proof}
Consider a repair scheme with repair vectors $\left( \mathbf{r}^{\ell}_j \right)^T, ~\ell \in [k+1,n],~ j \in [1,\beta]$. Taking an arbitrary reference vector $\left( \mathbf{r}\right)^T$, every other repair vector can be written in the form given by (\ref{Eqn:repairwrite}) using Lemma \ref{Lemma:multexist}. Consider the repair field elements $M_j^{\ell}$ that correspond to the multiplication operators $\mathbf{M}_j^{\ell}$ obtained from (\ref{Eqn:repairwrite}) . All multiplication matrices are full rank matrices. If there exists scalars $v^{\ell}_j \in \mathbb{GF}(p), ~\ell \in \{\ell_1,\ell_2 \dots \ell_q \}, j \in \{j_1,j_2 \dots j_q\}$, with at least one nonzero $v^{\ell}_j $, such that: 
\begin{equation}
  \sum \limits_{\ell,j} v^{\ell}_j \left(\mathbf{r}^{\ell}_j \right)^T \mathbf{P}^{(\ell)}_u  =0 .
\end{equation}
Then this is equivalent to
\begin{equation}
  \sum \limits_{\ell,j} v^{\ell}_j \left(\mathbf{r} \right)^T \mathbf{M}_j^{\ell}  \mathbf{P}^{(\ell)}_u =0. 
\end{equation}
Using the fact that the reference vector $\left(\mathbf{r}\right)^T $ is non zero and Property P1:
\begin{equation}
  \sum \limits_{\ell,j} v^{\ell}_j M_j^{\ell}  {P}^{(\ell)}_u =0. 
\end{equation}
This gives the rank condition over sub-field $\mathbb{GF}(p)$ as stated in (\ref{Eqn:rankcond}). This proves the forward direction.

For the converse, given a set of repair field elements it is possible to construct a set of repair multiplication operators and together with an arbitrary choice of a non-zero reference repair
vector, one can construct repair vectors satisfying the same rank conditions. 
\end{proof}

\subsection{Proof of Lemma \ref{Thm:HighLowSub}}
\begin{proof}
It is enough to show that $b$ field elements of the form $\{M_{j}P_{j}\},~1 \leq j \leq b$ are linearly dependent over $\mathbb{GF}(p^a)$  if and only if $ab$ field elements
$\{M_{j}P_{j}\zeta^{s}\},~1 \leq j \leq b,~0 \leq s \leq a-1$ are also linearly dependent over $\mathbb{GF}(p)$. Here, $M_{j}$ correspond to the repair field elements and $P_{j}$ correspond to the coefficients of the generator matrix corresponding to the parity node. Linear dependence over $\mathbb{GF}(p^a)$
implies that there exists scalars $v_j \in \mathbb{GF}(p^a)$, with at least one of them non-zero, such that $\sum \limits_{j} v_j M_{j}P_{j} = 0$.
Let us rewrite field elements $v_j$ in $\mathbb{GF}(p^a)$ as polynomials in $\zeta$ with coefficients $v_{js}$ from $\mathbb{GF}(p)$. Hence, the linear dependency relation becomes $\sum \limits_{j} \sum \limits_{s} v_{js}M_{j}P_{j}\zeta^s = 0$.
Hence over $\mathbb{GF}(p)$, $\{M_jP_{j}\zeta^s\},~ 1 \leq b,~ 0 \leq s \leq a-1$ are linearly dependent. The converse is also true since some scalar set 
$\zeta_{js} \in \mathbb{GF}(p)$ with one non zero element determines a scalar set $v_{j} \in \mathbb{GF}(p^a)$ with one non zero element. Hence, the claim follows.
\end{proof}
\end{appendix}

\subsection{Proof of Theorem \ref{Thm:WuDimakis}}
\begin{proof}
If $P^{(k+2)}_x \left(P^{(k+2)}_y\right)^{-1} \in  \mathbb{GF} \left(p^r\right)$ and $ P^{(k+2)}_y \left( P^{(k+2)}_z\right) ^{-1}  \in  \mathbb{GF}\left(p^r\right)$, then $P^{k+2}_x \left(P^{(k+2)}_y\right)^{-1}P^{k+2}_y \left(P^{k+2}_z\right)^{-1}  = \hspace{-10pt} P^{k+2}_x \left(P^{k+2}_z\right)^{-1} \in \mathbb{GF}(p^r)$.
The transitivity property partitions the graph $G$ into disjoint cliques. 

Using Theorem \ref{Thm:RepairField} we have that there are two repair field elements, i.e. $1$, $\mu \in \mathbb{GF}(p^{2r})$ corresponding to the two repair vectors that will be used to multiply the contents of the two parities respectively. 
The repair field elements for parity $1$ is $1$ because the corresponding repair vector acts as the reference vector. 
Then, the rank of $i$-th block  is $2$ if $1$ and $\mu P^{(k+2)}_i$ are linearly independent over sub-field $\mathbb{GF}(p^r)$. 
Similarly, the rank would be $1$ if they are linearly dependent, i.e $\mu P^{(k+2)}_i \in \mathbb{GF}(p^r)$. 
Now we establish the following property: if $i$ and $j$ are in the same clique, then either both columns of elements are simultaneously linearly dependent or linearly independent over $\mathbb{GF}(p^r)$. This is due to the fact that $\mu P^{k+2}_i \in \mathbb{GF}\left(p^r\right)$ forces $\mu P^{k+2}_i \left(P^{k+2}_i\right)^{-1} P^{k+2}_j \in \mathbb{GF}\left(p^r\right)$.

Similarly, if $i$ and $j$ are in different cliques, then the corresponding columns of vectors cannot be linearly dependent simultaneously. 
Suppose they are,
then $\mu\left(P^{(k+2)}_i\right) \in \mathbb{GF}(p^r)$ and  $\mu \left(P^{(k+2)}_j\right) \in \mathbb{GF}(p^r)$. Therefore, $\left(P^{(k+2)}_j\right)^{-1} \mu^{-1}\mu P^{(k+2)}_i = \left(P^{(k+2)}_j\right)^{-1}P^{(k+2)}_i \in \mathbb{GF}(p^r)$. But $(i,j) \notin E$ and therefore a contradiction.

For repair of node $i$, $\mu$ is chosen in such a way that $\mu P^{(k+2)}_i \notin \mathbb{GF}(p^r)$, so that the corresponding column of vectors are linearly independent. 
This selection of $\mu$ forces all blocks corresponding to the nodes in the same clique to be linearly independent and it can at most make columns corresponding to exactly one other clique
linearly dependent. Hence, the reduction in number of equations to be downloaded comes from the dependent clique. From this, the last claim in the theorem follows. 
\end{proof}

\end{document}